\documentclass[graybox]{svmult}

\usepackage{mathptmx}       
\usepackage{helvet}         
\usepackage{courier}        
\usepackage{type1cm}        
\usepackage{makeidx}         
\usepackage{graphicx}        
\usepackage{multicol}        
\usepackage[bottom]{footmisc}

\usepackage{amssymb}
\usepackage{amsmath}
\usepackage{bbm}

\newcommand{\ii}{\mathrm{i}}

\newcommand{\de}{\partial}

\newcommand{\magn}{\Delta_\mathbf{A}}
\newcommand{\magnh}{\text{H}}

\makeindex             

\begin{document}

\title*{Remarks on the derivation of Gross-Pitaevskii equation with magnetic Laplacian}
\authorrunning{Alessandro Olgiati}
\author{Alessandro Olgiati,\\$\,$ \footnotesize\\SISSA - International School for Advanced Studies,\\ via Bonomea 265, Trieste, Italy\\\email{aolgiati@sissa.it}}
%
%
\maketitle

\abstract*{The effective dynamics for a Bose-Einstein condensate in the regime of high dilution and subject to an external magnetic field is governed by a magnetic Gross-Pitaevskii equation. We elucidate the steps needed to adapt to the magnetic case the proof of the derivation of the Gross-Pitaevskii equation within the ``projection counting'' scheme.}

\abstract{The effective dynamics for a Bose-Einstein condensate in the regime of high dilution and subject to an external magnetic field is governed by a magnetic Gross-Pitaevskii equation. We elucidate the steps needed to adapt to the magnetic case the proof of the derivation of the Gross-Pitaevskii equation within the ``projection counting'' scheme.}

\section{Introduction and result}

The purpose of this note is to provide explicitly the non trivial adaptations of the known result \cite{p-external2015} which are needed to prove the derivation of the so-called time-dependent magnetic Gross-Pitaevskii equation from the many-body Schr\"odinger dynamics of a dilute gas of identical bosons subject to an external magnetic field. The presentation is therefore somewhat technical; nonetheless, since, to our knowledge, no explicit details were so far available in the literature, we propose it as a reference for the increasingly interesting topic of the effective many-body quantum dynamics with magnetic field.

The rigorous derivation of the Gross-Pitaevskii equation has been over the last two decades a central topic in the mathematics of the Bose gas; in its essence, it is a problem of persistence of condensation, or propagation of chaos, in the following sense. Suppose that the initial datum of a three dimensional Bose gas displays condensation onto a one-body state $u_0\in L^2(\mathbb{R}^3)$, namely
\[
\lim_{N\rightarrow\infty}\gamma^{(1)}_{N,0}=|u_0\rangle\langle u_0|,
\]
where $\gamma^{(1)}_{N,0}$ is the one-particle reduced density matrix associated to the initial datum $\psi_{N,0}$. Then condensation persists up to some time $T$ if 
\[
\lim_{N\rightarrow\infty}\gamma^{(1)}_{N,t}=|u_t\rangle\langle u_t|,\quad\forall t\in[0,T],
\]
for a condensate wave-function $u\equiv u_t(x)$ solution to the Gross-Pitaevskii equation
\[
\ii \de_t u= -\Delta u+8\pi a |u|^2 u
\]
with initial datum $u_0$. Here $a$ is the scattering length of the pair interaction among the particles of the many-body system.

The first complete proof of a result of this type is due to Erd\H{o}s, Schlein, and Yau in 2006 (see \cite{esy-2007} and \cite{esy-2010}); it was later reproduced with different methods by Pickl \cite{p-external2015}, by Benedikter, de Oliveira, and Schlein \cite{bdos-quant2012}, and by Brennecke and Schlein \cite{bs-2017}. All such derivations deal with a system of $N$ interacting bosons in the Gross-Pitaevskii scaling limit with non-relativistic kinetic operator given by $-\Delta$; this corresponds to a many-body Hamiltonian of the form
\[
H_N=\sum_{i=1}^N(-\Delta_i)+\sum_{i<j}N^2V(N(x_i-x_j)).
\]
Such methods can be adapted if the one-body Laplacian is modified by the insertion of an external (confining) potential. Analogously, it is of great relevance and interest to insert an external magnetic field which the charged particles are coupled with; mathematically this is modeled, with minimal coupling, by replacing the kinetic part in $H_N$ with its magnetic counterpart
\[
\sum_{i=1}^N(-\Delta_\mathbf{A})_i:=\sum_{i=1}^N(-\ii\nabla_i+\mathbf{A}(x_i))^2,
\]
where $\mathbf{A}:\mathbb{R}^3\rightarrow\mathbb{R}^3$ is a vector potential. This would in turn imply the effective dynamics to be ruled by the magnetic Gross-Pitaevskii equation
\begin{equation} \label{eq:magneticGP}
\ii \de_tu_t= -\Delta_\mathbf{A} u_t+8\pi a |u_t|^2 u_t.
\end{equation}
The fact that an external magnetic field can be accommodated into the many-body dynamics, and that the one-body marginal can be controlled analogously to what is done when the one-particle operator is simply the negative Laplacian, is to be expected and indeed is mentioned explicitly in \cite[Remark 2.1]{p-external2015}. However, such an adaptation is not as straightforward as the analogous insertion of an external trapping potential: the magnetic Laplacian is formally the sum of the ordinary Laplacian plus a derivative term that is linear in the magnetic potential and a further quadratic term in the magnetic potential itself; this more complicated structure requires an a priori not immediate adjustment of a number of crucial estimates and steps in the main proof. For the related problem of derivation of the magnetic Hartree equation from many-body quantum dynamics, the reader should refer to \cite{l-2012}.

Before stating the result, let us define the magnetic Sobolev space $\text{H}^k_\mathbf{A}$ as the set of $u\in L^2$ such that
\[
\|u\|_{\text{H}^k_\mathbf{A}}^2=\sum_{0\le j\le k}\|(\nabla-\ii\mathbf{A})^ju\|_2^2<+\infty.
\]
We will consider the magnetic Hamiltonian
\[
H_{N,\mathbf{A}}:=-\sum_{i=1}^N\Delta_{i,\mathbf{A}}+\sum_{i<j}N^2V(N(x_i-x_j)),
\]
as the generator of the linear many-body Schr\"odinger dynamics. Moreover, we define the two $\mathbf{A}$-dependent energy functionals
\begin{equation} \label{eq:manybodyenergy}
\mathcal{E}_N(\psi_N):=\frac{1}{N}\langle\psi_N,H_{N,\mathbf{A}}\psi_N\rangle
\end{equation}
and
\begin{equation} \label{eq:gpenergy}
\mathcal{E}^{GP}(u):=\langle u,-\Delta_{\mathbf{A}}u\rangle+ 4\pi a \langle u, |u|^2 u \rangle.
\end{equation}
They represent the energies conserved along the flow of, respectively, the many-body Schr\"odinger equation and the magnetic Gross-Pitaevskii equation. We can now state the result as follows.

\begin{theorem} \label{thm:main}
	Let $V$ be a positive, $L^\infty$, spherically symmetric, and compactly supported function on $\mathbb{R}^3$, and let $\mathbf{A}\in W^{1,\infty}(\mathbb{R}^3,\mathbb{R}^3)$ be chosen such that $\nabla\cdot\mathbf{A}=0$.
	Suppose that the sequence of initial many-body states $\{\psi_{N,0}\}_{N\in\mathbb{N}}$ is condensed in the sense of reduced densities, i.e.,
	\[
	\lim_{N\rightarrow\infty}\gamma^{(1)}_{N,0}=|u_0\rangle\langle u_0|
	\]
	on a condensate wave-function $u_0\in {\text{\normalfont H}}^2_{\mathbf{A}}$ (here $\gamma^{(1)}_{N,0}$ is the one-particle reduced density matrix of $\psi_{N,0}$). Suppose in addition that
	\[
	\lim_{N\rightarrow\infty}\mathcal{E}_N(\psi_{N,0})=\mathcal{E}^{GP}(u_0).
	\]
	Then one has condensation for all $t>0$, that is
	\begin{equation} \label{eq:thesis}
	\lim_{N\rightarrow\infty}\gamma_{N,t}^{(1)}=|u_t\rangle\langle u_t|
	\end{equation}
	on a state $u_t$ that solves the magnetic Gross-Pitaevskii equation \eqref{eq:magneticGP} with initial datum $u_0$. Here $a$ is the scattering length of the interaction $V$.	
\end{theorem}

We remark that our hypotheses on $\mathbf{A}$ certainly ensures that $\|\cdot\|_{\text{H}^k_{\mathbf{A}}}$ is equivalent to the standard Sobolev norm $\|\cdot\|_{\text{H}^k}$ for $k\in\{0,1,2\}$; indeed, for any $f\in \text{H}^2$, one has
\[
\|\Delta_\mathbf{A} f\|_2\lesssim \|\Delta f\|_2+\|\mathbf{A}\|_\infty\|\nabla f\|_2+\|\mathbf{A}\|_\infty^2\|f\|_2\lesssim \|f\|_{\text{H}^2}
\]
and, for any $f\in \magnh^2_\mathbf{A}$,
\[
\begin{split}
\|\Delta f\|_2\lesssim&\|\Delta_\mathbf{A}f\|_2+\|\mathbf{A}\|_\infty\|\nabla f\|_2+\|\mathbf{A}\|_\infty^2\|f\|_2.
\end{split}
\]
Since $\|\nabla f\|_2\lesssim \epsilon \|\Delta f\|_2+1/\epsilon\|f\|_2$ for any $\epsilon>0$, by choosing $\epsilon>0$ small enough one gets $\|f\|_{\text{H}^2}\lesssim\|f\|_{\text{H}^2_\mathbf{A}}$. The cases $k=0$ and $k=1$ follow trivially.

We also stress that, again due to the hypotheses $\mathbf{A}\in W^{1,\infty}$ and $\nabla\cdot\mathbf{A}=0$, the global existence of solution to the magnetic Gross-Pitaevskii equation \eqref{eq:magneticGP} in the magnetic Sobolev spaces up to $k=2$ is granted due to standard arguments. It would be of great interest to find a larger class of vector potentials such that a result similar to Theorem \ref{thm:main} holds: for example, a constant magnetic field $\mathbf{B}=\nabla\times \mathbf{A}$ is not attainable by $\mathbf{A}\in W^{1,\infty}$.

An interesting future outlook is the derivation of the magnetic Gross-Pitaevskii equation  for time-dependent magnetic potentials $\mathbf{A}(t)$. Since the treatment in \cite{p-external2015} already deals with time-dependent external (electric) fields, it is expected that such result could be extended to cover a suitable class of $\mathbf{A}(t)$ having enough space and time regularity.

\section{Proof of Theorem \ref{thm:main}}
Theorem \ref{thm:main} is proven with the same strategy as Theorem 2.1 in \cite{p-external2015}. The crucial quantity one wants to control is
\begin{equation} \label{eq:alpha}
\alpha_{N,t}:=\langle\psi_N,\widehat{m}\psi_N\rangle+|\mathcal{E}_N(\psi_N)-\mathcal{E}^{GP}(u)|-N(N-1)\text{Re}\langle\psi_N,g_\beta(x_1-x_2)\widehat{r}\psi_N\rangle.
\end{equation}
For the definition of $\widehat{m}$ and $\widehat{r}$ in \eqref{eq:alpha} see \cite[Def. 6.1 and Def. 6.2]{p-external2015}. The definition of $g_\beta$ is recalled in eq. \eqref{eq:def_g}, since its role is slightly modified by the presence of $\mathbf{A}$. The core of the proof is to look for an estimate of the form
\begin{equation} \label{eq:gronwall}
\partial_t\alpha_{N,t}\le C(t)\Big(\langle\psi_N,\widehat{m}\psi_N\rangle+|\mathcal{E}_N(\psi_N)-\mathcal{E}^{GP}(u)| +N^{-\eta}\Big)
\end{equation}
for some $\eta>0$. By Gr\"onwall Lemma, this is enough to get \eqref{eq:thesis} (see \cite[Sect. 6]{p-external2015} for details). The factor $C(t)$, which varies from step to step during the proof, represents a function depending on the magnetic Sobolev norms  $\|\psi_{N,t}\|_{\text{H}_\mathbf{A}^1}$ and $\|u_t\|_{\text{H}_\mathbf{A}^2}$; for this reason, it is in general exponentially growing in time, but not $N$-dependent.

Computing the time-derivative of $\alpha_{N,t}$ one gets 
\begin{equation}
\de_t\alpha_{N,t}\le \gamma_b+\gamma_c+\gamma_d+\gamma_e+\gamma_f+\gamma_l,
\end{equation}
where the terms $\gamma_j$, $j\in\{b,c,d,e,f\}$ are defined in \cite[Def. 6.6]{jlp-two2016} and \cite[Def. 6.3]{p-external2015}, while the new summand
\begin{equation}\label{eq:gammal}
\gamma_l:=N^2\big|\langle\psi_N,\,\nabla_{x_1}g_\beta(x_1-x_2)\,\mathbf{A}(x_1) \widehat{r}\,\psi_N \rangle\big|
\end{equation}	
emerges in our case due to the presence of $\mathbf{A}$; let us remark that for us $\gamma_a=0$ since we are not considering external traps.

In \cite[Appendix A.2]{p-external2015} it is shown in detail how $\gamma_j$, $j\in\{b,c,d,e,f\}$ (see \cite[Sect. 6.4]{jlp-two2016} for the estimate of $\gamma_f$) can be bounded in terms of $\langle\psi_N,\widehat{m}\psi_N\rangle$, $|\mathcal{E}_N(\psi_N)-\mathcal{E}^{GP}(u)|$ and $N^{-\eta}$, in order to obtain \eqref{eq:gronwall}. We report in what follows the main adaptations needed in the magnetic case for the treatment presented in \cite[Appendix A.2]{p-external2015}, plus the estimate of the additional term $\gamma_l$.

\subsection{Cancellation of the kinetic part}
A remarkable feature of the counting method we are considering here (introduced in \cite{p-simple2011} and \cite{kp-singular2010}) is that the single-particle terms in $H_N$ (among them the kinetic part) get canceled exactly when computing $\de_t \alpha_{N,t}$; in \cite{p-external2015}, this happens in Lemma 6.2 and it occurs in the case of $-\Delta_\mathbf{A}$ as well. More precisely, when computing $\de_t\langle\psi_N,\widehat{m}\psi_N\rangle$, one has
\[
\de_t\langle\psi_N,\widehat{m}\psi_N\rangle=\ii\Big\langle\psi_N,\Big[H_{N,\mathbf{A}}-\sum_{i=1}^N(-\Delta_{\mathbf{A},x_i}+8\pi a |u|^2_i),\widehat{m}\Big]\psi_N\Big\rangle,
\]
and one easily sees that the magnetic Laplacians get exactly canceled. This cancellation is the reason why, in the less involved mean-field case considered in \cite{kp-singular2010}, not much needs be done to deal with magnetic Laplacians. Apart from technical assumptions, all the proof proceeds in the same way since $-\magn$ does not play a role. In the Gross-Pitaevskii regime however, even though the cancellation takes place and the kinetic part does not have to be directly estimated, nonetheless $-\magn$ still plays a role along the proof through the emergence of the energy difference $|\mathcal{E}_N(\psi_N)-\mathcal{E}^{GP}(u)|$.

\subsection{Cancellation of $V_N-W_\beta$} \label{subsect:scattering}
In analogy to the other known derivations of the Gross-Pitaevskii equation, one needs to include in the treatment a function displaying some short-scale structure that allows one to weaken the strong singularity of the interaction term $N^2V(N\cdot)$. This is done by means of the solution $f_\beta$ to the zero-energy scattering problem relative to the modified potential $V_N-W_\beta$, where $W_\beta$ is the less singular potential introduced in \cite[Sect. 5]{p-external2015} so as to make $V_N-W_\beta$ have zero scattering length. $f_\beta$ is thus the solution to
\begin{equation}\label{eq:scattering}
\Big(-\Delta+\frac{1}{2}(V_N-W_\beta)\Big)f_\beta=0,
\end{equation}
with $f_\beta\rightarrow1$ for $|x|\rightarrow\infty$. The function $g_\beta$ that appears in \eqref{eq:alpha} is defined as
\begin{equation} \label{eq:def_g}
g_\beta:=1-f_\beta.
\end{equation}

As explained in \cite[Sect. 6.2]{p-external2015}, the function $g_\beta$ plays a crucial role in the replacement of the strong potential $V_N$, which is of order $N^2$ at short distances, with the softer $W_\beta$, which is instead of order $N^{3\beta-1}$; this is of course at the expense of the appearance of their difference, but this can be shown to disappear exactly modulo terms that can be estimated. Performing all calculations for $\de_t\alpha_{N,t}$ in the magnetic case, one gets as already mentioned the terms $\gamma_b$ to $\gamma_f$ as appearing in \cite[Def. 6.3]{p-external2015} and \cite[Def 6.6]{jlp-two2016}; however, when computing $[H_N,g_\beta(x_1-x_2)]$ as one can find after \cite[Eq. 6.17]{p-external2015}, one gets
\begin{equation}
\begin{split}
[H_N,g_\beta(x_1-x_2)]=&[\Delta_{\mathbf{A},x_1}+\Delta_{\mathbf{A},x_2}, f_\beta(x_1-x_2)]\\
=&(V_N-W_\beta)f_\beta(x_1-x_2)-2(\nabla_{x_1} g_\beta(x_1-x_2))\nabla_{x_1}\\
&-2(\nabla_{x_2} g_\beta(x_1-x_2))\nabla_{x_2}-2\ii\mathbf{A}(x_1)(\nabla_{x_1}g_\beta(x_1-x_2))\\
&-2\ii\mathbf{A}(x_2)(\nabla_{x_2}g_\beta(x_1-x_2)),
\end{split}
\end{equation}
having used $\nabla\cdot\mathbf{A}=0$ and \eqref{eq:scattering}. The terms containing $(\nabla g_\beta)\nabla$ are present in \cite{p-external2015} too, and they provide the term $\gamma_c$. The terms containing $\mathbf{A}$ were instead not present in the purely kinetic case, and they exactly correspond to $\gamma_l$.

\subsection{Adapting the estimates}

To get the desired estimate \eqref{eq:gronwall} one has to treat separately $\gamma_b$, $\gamma_c$, $\gamma_d$, $\gamma_e$, $\gamma_f$, $\gamma_l$. The calculations proceed exactly as in \cite[Appendix A.2]{p-external2015}, with some modifications we describe here.

\subsubsection{Insertion of $h_{\beta_1,\beta}$}

Lemma A.4 in \cite{p-external2015} is used to prove the bound for $\gamma_b$ and in its proof (to treat the term of type III for small $\beta$ and of type I, II and III for arbitrary $\beta$) one replaces $V_\beta$ with $U_{\beta_1,\beta}+\Delta h_{\beta_1,\beta}$; for example, one has (see \cite[proof of Lemma A.4 (3), for $\beta$ small]{p-external2015})
\[
\begin{split}
N^2\big|\langle\psi_N, q_1p_2V_\beta(x_1-x_2)\widehat {m} q_1 q_2\psi_N\rangle\big|\le&\,N^2\big|\langle\psi_N, q_1p_2U_{0,\beta}(x_1-x_2)\widehat {m} q_1 q_2\psi_N\rangle\big| \\
&+N^2\big|\langle\psi_N, q_1p_2(\Delta_1h_{0,\beta}(x_1-x_2))\widehat {m} q_1 q_2\psi_N\rangle\big|
\end{split}
\]
The first summand can be bounded easily, since $U_{0,\beta}$ is less singular than $V_\beta$. To treat the second summand, the strategy is then to integrate by parts $\Delta h_{\beta_1,\beta}$ once or twice and then to manipulate the outcome in order to obtain the Sobolev norms of $\Psi_{N,t}$ or $u_t$. This procedure can be adapted to the magnetic case since one can use the trivial relation
\[
\nabla=\nabla_\mathbf{A}+\ii\mathbf{A},
\]
which allows to get a magnetic gradient at the expense of a $L^\infty$-bounded term.
This allows to bound the second summand by
\begin{align}
&N^2\big|\langle\nabla_{1,\mathbf{A}}q_1p_2\psi_N, (\nabla_{1} h_{0,\beta}(x_1-x_2))\widehat {m} q_1 q_2\psi_N\rangle\big| \label{eq:third_term_1}\\
&+ N^2\big|\langle\psi_N, q_1p_2(\nabla_{1} h_{0,\beta}(x_1-x_2))\nabla_{1,\mathbf{A}}\widehat {m} q_1 q_2\psi_N\rangle\big|\label{eq:third_term_2}\\
&+N^2\big|\langle\psi_N, q_1p_2\mathbf{A}(x_1)(\nabla_{1} h_{0,\beta}(x_1-x_2))\widehat {m} q_1 q_2\psi_N\rangle\big|. \label{eq:third_term_3}
\end{align}
At this point one can repeat the computations performed in \cite{p-external2015} to bound the terms (A.14) to (A.17), the only difference being that $\nabla_\mathbf{A}$ will produce magnetic norms in the estimates of \eqref{eq:third_term_1} and \eqref{eq:third_term_2}; \eqref{eq:third_term_3} is even less singular, since it contains only one derivative, and it can again be bounded by repeating the bounds for \cite[Eq. A.14 to A.17]{p-external2015}.

\subsubsection{Magnetic norms}

The Sobolev norms $\|\psi_{N,t}\|_{\magnh^1}$ or $\|u_t\|_{\magnh^k}$ with $k=1,2$ emerge frequently along the proof, not only due to the integration by parts of $\Delta h_{\beta_1,\beta}$, but also typically by a Sobolev embedding argument (see e.g. \cite[Eq. A.37 and A.15]{p-external2015}), or due to \cite[Prop. A.3]{p-external2015}. While in the non-magnetic case, such terms are bounded by some $N$-independent function of time, in the case of $\mathbf{A}\ne0$ one needs to use the inequality $\|\cdot\|_{\magnh^k}\leqslant\,C\,\|\cdot\|_{\magnh^k_\mathbf{A}}$ granted by the equivalence of the two norms for $k=1,2$. Then, by general facts about magnetic Schr\"odinger equations, the two norms $\|\psi_{N,t}\|_{\magnh^1_\mathbf{A}}$ and $\|u_t\|_{\magnh^1_\mathbf{A}}$ are uniformly bounded in time. The magnetic Sobolev norm $\|\cdot\|_{\text{H}_\mathbf{A}^2}$ is instead not a priori bounded, but the $W^{1,\infty}$-boundedness of $\mathbf{A}$ allows to get
\[
\|u_t\|_{\text{H}_\mathbf{A}^2}\le D e^{K|t|},
\] in the same way as for the non-magnetic case. The norm $\| u_t\|_\infty$ often appears as well, typically every time \cite[Lemma 4.1 (5)]{p-external2015} is used; $\| u_t\|_\infty$ can of course be bounded by $\|u_t\|_{\text{H}^2}$ by standard embedding arguments, and hence by $C\,\|u_t\|_{\text{H}^2_\mathbf{A}}$ again by equivalence of norms.

\subsubsection{Lemma 5.2 of \cite{p-external2015}}

Lemma 5.2 in \cite{p-external2015} allows one to bound a part of the kinetic energy by means of the functional $\alpha_{N,t}$ and $N^{-\eta}$; it plays a role in the estimate of the term of type III in Lemma A.4 of \cite{p-external2015} and in the bound of $\gamma_d$ \cite[pages 39 through 41]{p-external2015}. It still holds in our case, with the substitution $\nabla\mapsto \nabla_{\mathbf{A}}$ and with the appropriate magnetic energy functionals defined in \eqref{eq:manybodyenergy} and \eqref{eq:gpenergy}. In the proof (see \cite[Appendix A.3]{p-external2015}), one has exactly all the magnetic analogous of the terms \cite[Eqs. A.53 to A.60]{p-external2015}.
The term corresponding to \cite[Eq. A.54]{p-external2015} can be bounded by
\[
\begin{split}
|\langle\nabla_{1,\mathbf{A}}q_1\psi_N,\mathbb{I}_{\mathcal{A}_1}\nabla_{1,\mathbf{A}}p_1\psi_N\rangle| \le& \, |\langle\nabla_{1,\mathbf{A}}q_1\psi_N,\nabla_{1,\mathbf{A}}p_1\psi_N\rangle| \\&+|\langle\nabla_{1,\mathbf{A}}q_1\psi_N,\mathbb{I}_{\overline{\mathcal{A}_1}}\nabla_{1,\mathbf{A}}p_1\psi_N\rangle|  \\
\le&\, |\langle \widehat{n}^{-1/2}q_1\psi_N,\Delta_{1,\mathbf{A}}\widehat{n}^{1/2}_1p_1\psi_N\rangle|\\
&+\|\mathbb{I}_{\overline{\mathcal{A}}_1}\|_{op} \|\nabla_{1,\mathbf{A}}q_1\psi_N\|\|\nabla_{1,\mathbf{A}} p_1\|_{op} \\
\le &\,C(t) \Big(\langle\psi_N,\widehat{n}\psi_N\rangle+ N^{-\eta}\Big),
\end{split}
\]
having used \cite[Lemma 4.1 (3)]{p-external2015} as well as the fact that $\widehat{n}^{-1/2}$ is well defined on $\text{Ran}\,q_1$ for the second step and \cite[Prop. A.1 (2)]{p-external2015} for the third one. Here $\mathbb{I}_{\mathcal{A}_1}$ is the characteristic function of the set $\mathcal{A}_1$ defined in \cite[Def. 5.2]{p-external2015}, while $C(t)$ is a function depending on the magnetic Sobolev norm $\|u_t\|_{\text{H}_\mathbf{A}^2}$. With similar arguments one can bound the magnetic analogous of \cite[Eq. A.59]{p-external2015}, i.e.,
\[
\|\mathbb{I}_{\mathcal{A}_1}\nabla_{1,\mathbf{A}}p_1\psi_N\|^2-\|\nabla_{1,\mathbf{A}}u\|^2,
\]
and this is enough to get the thesis of \cite[Lemma 5.2]{p-external2015} (the interaction terms are of course unmodified by the insertion of $\mathbf{A}$).

\subsubsection{Bound on $\gamma_l$}
We show here how the term $\gamma_l$ defined in \eqref{eq:gammal} can be estimated in order to get \eqref{eq:gronwall}.
\begin{lemma} \label{lemma:gamma}
	There exists $\eta>0$ such that
	\[
	\gamma_l\leqslant C(t)\, N^{-\eta}
	\]
	for a function $C(t)$ depending on $\|u_t\|_{\text{H}_\mathbf{A}^2}$ but not on $N$.
\end{lemma}

\begin{proof}
	We recall that
	\[
	\widehat{r}:=p_1p_2\widehat{m}^{\,b}+(p_1q_2+q_1p_2)\widehat{m}^{\,a},
	\]
	where $\widehat{m}^{\,b}$ and $\widehat{m}^{\,a}$ are in \cite[Def. 6.2]{p-external2015}. By symmetry of $g_\beta$, we can integrate by parts in the $x_2$ variable; we get
	\begin{equation}
	\begin{split} \label{eq:gammaestimate}
	|\gamma_l|\leqslant& N^2\big|\langle\nabla_{x_2}\psi_N,\,g_\beta(x_1-x_2)\,\mathbf{A}(x_1) \widehat{r}\,\psi_N \rangle\big|\\&+N^2\big|\langle\psi_N,\,g_\beta(x_1-x_2)\,\mathbf{A}(x_1)\nabla_{x_2} \widehat{r}\,\psi_N \rangle\big|.
	\end{split}
	\end{equation}
	We can use the definition of $\widehat{r}$ for the first term and get
	\[
	N^2\big|\langle\nabla_{x_2}\psi_N,\,g_{12}\,\mathbf{A}(x_1) \widehat{r}\,\psi_N \rangle\big|\le N^2\|\nabla_2\psi_N\| \|\mathbf{A}\|_\infty\|g_{12} p_1\|_\infty (\|\widehat{m}^{\,a}\|_{op}+\|\widehat{m}^{\,b}\|_{op}),
	\]
	having used the short-hand notation $g_{12}:=g_\beta(x_1-x_2)$. Now, by \cite[Lemma 4.1]{p-external2015}, \cite[Lemma 5.1]{p-external2015} and \cite[Eq 6.11]{p-external2015}, one gets
	\[
	N^2\big|\langle\nabla_{x_2}\psi_N,\,g_{12}\,\mathbf{A}(x_1) \widehat{r}\,\psi_N \rangle\big|\leqslant C(t)\, N^{1+\xi}\|\psi_N\|_{H^1_\mathbf{A}}\|g_\beta\|\leqslant C(t)N^{-\beta/2+\xi},
	\]
	for some $\xi>0$ to be chosen suitably small. Here we used the uniform boundedness of the first magnetic Sobolev norm $\|\psi_N\|_{\text{H}_\mathbf{A}^1}$ and the fact that $\|u_t\|_\infty$, produced by \cite[Lemma 4.1]{p-external2015}, is bounded by $C\,\|u_t\|_{\text{H}_\mathbf{A}^2}$.
	
	As for the second term in \eqref{eq:gammaestimate}, we can remark that two summands of $\widehat{r}$ contain $p_1$, and their sum is equal to $p_1\widehat{r}$. For them, one can use H\"older inequality in the variable $x_2$ and then Sobolev inequality again in the variable $x_2$ to get
	\[
	\begin{split}
	N^2\big|\langle\psi_N,\,g_{12}\,\mathbf{A}(x_1)&\nabla_{x_2} \,p_1\,\widehat{r}\,\psi_N \rangle\big|\leqslant  \,N^2 \int d^3x_1\,d^3 x_3\dots d^3x_N\|g_\beta(x_1-\cdot)\|_{3/2}\\
	& \times\|\psi_N(x_1,\cdot,x_3\dots x_N)\|_6\|\mathbf{A}(x_1)(\nabla\,p_1\,\widehat{r}\,\psi_N)(x_1,\,\cdot\,,x_3\dots x_N)\|_6\\
	&\qquad\qquad\quad\leqslant  \,N^2 \|g_\beta\|_{3/2}\|\mathbf{A}\|_\infty\int d^3x_1\,d^3 x_3\dots d^3x_N \\
	&\times \|\nabla\psi_N(x_1,\,\cdot\,,x_3\dots x_N)\|\|(\Delta\,p_1\,\widehat{r}\,\psi_N)(x_1,\,\cdot\,,x_3\dots x_N)\|\\
	\leqslant & \,C(t)\,N^2 \|\psi_N\|_{H^1_\mathbf{A}}\|g_\beta\|_{3/2}\|\Delta u\|(\|\widehat{m}^{\,a}\|_{op}+\|\widehat{m}^{\,b}\|_{op}),
	\end{split}
	\]
	having used in the last step the definition of $\widehat{r}$, the fact that $\|\Delta p\|_{op}=\|\Delta u\|_2$ and \cite[Cor. 4.1]{p-external2015}. By interchanging the roles of $x_1$ and $x_2$, the same estimate can be proven if $q_1\widehat{r}$ replaces $p_1\widehat{r}$. One can now use $\|\Delta u\|\leqslant C\|u\|_{H^2_{\mathbf{A}}}$, \cite[Lemma 5.1]{p-external2015} (plus a standard interpolation argument to obtain $\|g_\beta\|_{3/2}\le \|g_\beta\|_2^{2/3}\,\|g_\beta\|_1^{1/3}\le C\,N^{-1-\beta_1}$) and \cite[Eq. 6.11]{p-external2015} and get
	\[
	N^2\big|\langle\psi_N,\,g_{12}\,\mathbf{A}(x_1)\nabla_{x_2} \widehat{r}\,\psi_N \rangle\big|\leqslant C(t)\, N^{-\beta+\xi},
	\]
	which is enough to get the thesis.
\end{proof}

\begin{acknowledgement}
{Partially supported by the 2014-2017 MIUR-FIR grant ``\emph{Cond-Math: Condensed Matter and Mathematical Physics}'', code RBFR13WAET and by Gruppo Nazionale per la Fisica Matematica (GNFM-INdAM). The author also warmly thanks the GSSI, for the kind hospitality and financial support during a visit in L'Aquila.}
\end{acknowledgement}
%

\end{document}